\documentclass[journal]{IEEEtran}
\usepackage{amssymb,url}
\usepackage{graphicx,epsfig,color,wrapfig}

\def\mytitle{The Burbea-Rao and Bhattacharyya centroids}

\def\jMEF{{\sc jMEF}}
\def\dF{\mathrm{d}F}

\def\equaldef{\stackrel{\mathrm{\footnotesize def}}{=}}

\title{\mytitle}

\author{Frank~Nielsen,~\IEEEmembership{Senior Member,~IEEE,}  and     Sylvain~Boltz,~\IEEEmembership{Nonmember,~IEEE}    
\thanks{F. Nielsen is with the Department
of Fundamental Research of Sony Computer Science Laboratories, Inc., Tokyo, Japan, and the Computer Science Department (LIX) of \'Ecole Polytechnique, Palaiseau, France. 
 e-mail: Frank.Nielsen@acm.org}
\thanks{S. Boltz is with the Computer Science Department (LIX) of \'Ecole Polytechnique, Palaiseau, France. e-mail: boltz@lix.polytechnique.fr}
\thanks{Manuscript received April 2010, revised April 2012. This revision includes in the appendix a proof of the uniqueness of the centroid.}
}

\def\dx{\mathrm{d}x}

\def\innerproduct#1#2{ \langle {#1},{#2} \rangle }
\def\Innerproduct#1#2{ \left\langle {#1},{#2} \right\rangle }
\def\BR{\mathrm{BR}}
\def\tr#1{\mathrm{tr}(#1)}
\def\JS{\mathrm{JS}}
\def\KL{\mathrm{KL}}
\def\E{\mathcal{E}}

\newtheorem{theorem}{Theorem}
\newtheorem{corollary}{Corollary}
\newtheorem{lemma}{Lemma}

\begin{document}
\maketitle

\begin{abstract}
We study the centroid with respect to the class of information-theoretic Burbea-Rao divergences that generalize the celebrated Jensen-Shannon divergence by measuring the non-negative Jensen difference induced by a strictly convex and differentiable function.
Although those Burbea-Rao divergences are symmetric by construction, they are not metric since they fail to satisfy the triangle inequality.
We first explain how a particular symmetrization of Bregman divergences called Jensen-Bregman distances yields exactly those Burbea-Rao divergences. 
We then proceed by defining skew Burbea-Rao divergences, and show that skew Burbea-Rao divergences amount in limit cases to compute Bregman divergences.
We then prove that Burbea-Rao centroids are unique, and can be arbitrarily finely approximated by a generic iterative concave-convex optimization algorithm with guaranteed convergence property.
In the second part of the paper, we consider the Bhattacharyya distance that is commonly used to measure overlapping degree of probability distributions.
We show that Bhattacharyya distances on members of the same statistical exponential family amount to calculate a Burbea-Rao divergence in disguise.
Thus we get an efficient algorithm for computing the Bhattacharyya centroid of a set of parametric distributions belonging to the same exponential families, improving over former specialized methods found in the literature that were limited to univariate or ``diagonal'' multivariate Gaussians.
To illustrate the performance of our Bhattacharyya/Burbea-Rao centroid algorithm, we present experimental performance results for  $k$-means and hierarchical clustering methods of Gaussian mixture models.  
\end{abstract}


\IEEEpeerreviewmaketitle

\def\P{\mathcal{P}}

\section{Introduction} 

\subsection{Means and centroids}
In Euclidean geometry, the centroid $c$ of a point set $\P=\{p_1, ..., p_n\}$ is defined as the center of mass $\frac{1}{n}\sum_{i=1}^n p_i$, also characterized as the center point that minimizes the {\it average squared} Euclidean distances: $c=\arg\min_p \sum_{i=1}^n \frac{1}{n} \|p-p_i\|^2$.
This basic notion of Euclidean centroid can be extended to denote a {\it mean} point $M(\P)$ representing the {\it centrality} of a given  point set $\P$.
There are basically two complementary approaches to define mean values of numbers: (1) by axiomatization, or (2) by optimization, summarized concisely as follows:

\begin{itemize}

\item {\bf By axiomatization}. 
This approach was first historically pioneered by the independent work of Kolmogorov~\cite{mean-kolmogorov-1930} and Nagumo~\cite{mean-nagumo-1930} in 1930, and simplified and refined later by Acz\'el~\cite{meanvalueaxiomatization-1948}.
Without loss of generality we consider the mean of two non-negative numbers $x_1$ and $x_2$, and postulate the following expected behaviors  of a mean function $M(x_1,x_2)$  as axioms (common sense): 
\begin{itemize}
\item Reflexivity.  $M(x, x)=x$, 
\item Symmetry. $M(x_1,x_2)=M(x_2,x_1)$,  
\item Continuity and strict monotonicity.   $M(\cdot,\cdot)$ continuous and $M(x_1,x_2)<M(x_1',x_2)$ for $x_1<x_1'$, and
\item Anonymity.  $M(M(x_{11},x_{12}),M(x_{21},x_{22})) = M(M(x_{11},x_{21}),M(x_{12},x_{22}))$  (also called bisymmetry expressing the fact that  the mean can be computed as a mean on the row means or equivalently  as a mean on the column means).
\end{itemize}

Then one can show that the mean function $M(\cdot,\cdot)$ is necessarily written  as:
\begin{equation}
M(x_1,x_2) = f^{-1}\left(\frac{f(x_1)+f(x_2)}{2}\right) \equaldef M_f(x_1,x_2),
\end{equation}
for a strictly increasing function $f$.
The arithmetic $\frac{x_1+x_2}{2}$, geometric $\sqrt{x_1x_2}$ and harmonic means $\frac{2}{\frac{1}{x_1}+\frac{1}{x_2}}$ 
are instances of such generalized means obtained for $f(x)=x$, $f(x)=\log x$ and $f(x)=\frac{1}{x}$, respectively.
Those generalized means are also called {\it quasi-arithmetic means}, since they can be interpreted as the arithmetic mean on the sequence $f(x_1), ..., f(x_n)$, the $f$-representation of numbers.
To get geometric centroids, we simply consider means on each coordinate axis independently.
The Euclidean centroid is thus interpreted as the Euclidean arithmetic mean. 
Barycenters (weighted centroids) are similarly obtained using non-negative weights (normalized  so that $\sum_{i=1}^n w_i=1$):

\begin{equation}
M_f(x_1, ..., x_n; w_1, ..., w_n) = f^{-1}\left( 
\sum_{i=1}^n w_i f(x_i)
\right)
\end{equation}

Those generalized means satisfy the inequality property: 
\begin{equation}\label{eq:dom}
M_f(x_1, ..., x_n; w_1, ..., w_n) \leq  M_g(x_1, ..., x_n; w_1, ..., w_n),
\end{equation} if and only if function $g$ dominates $f$: That is, $\forall x, g(x)>f(x)$.
Therefore the arithmetic mean ($f(x)=x$) dominates the geometric mean ($f(x)=\log x$) which in turn dominates the harmonic mean $f(x)=\frac{1}{x}$. 
Note that it is {\it not} a strict inequality in Eq.~\ref{eq:dom} as the means coincide for all identical elements: 
if all $x_i$ are equal to $x$ then $M_f(x_1, ..., x_n)=f^{-1}(f(x))=x=g^{-1}(g(x))=M_g(x_1, ..., x_n)$.
All those quasi-arithmetic means further satisfy the ``interness'' property 

\begin{equation}
\min(x_1, ..., x_n) \leq M_f(x_1, ..., x_n) \leq \max(x_1, ..., x_n),
\end{equation}
derived from limit cases $p\rightarrow\pm\infty$ of power means\footnote{Besides the min/max operators interpreted as extremal power means, the geometric mean itself can also be interpreted as a power mean  
$(\prod_{i=1}^n x_i^p)^{\frac{1}{p}}$ in the limit case $p\rightarrow 0$. } for $f(x)=x^p, p\in\mathbb{R}_*=(-\infty,\infty)\backslash\{0\}$, a non-zero real number.

\item {\bf By optimization}. In this second alternative approach, the barycenter $c$ is defined according to a distance function $d(\cdot,\cdot)$ as the optimal solution of a minimization problem

\begin{equation}
\mathrm{(OPT)}: \min_{x} \sum_{i=1}^n w_i d(x,p_i) = \min_{x} L(x; \P, d),
\end{equation}

where the non-negative weights $w_i$ denote multiplicity or relative importance of points (by default, the centroid is defined by fixing all  $w_i=\frac{1}{n}$).
Ben-Tal et al.~\cite{entropicmeans-1989} considered an information-theoretic class of distances called $f$-divergences~\cite{AliSilvey-1966,Csiszar-1967}:

\begin{equation}
I_f(x,p) = p f\left( \frac{x}{p} \right),
\end{equation}
for a strictly convex differentiable function $f(\cdot)$ satisfying $f(1)=0$ and $f'(1)=0$.
Although those $f$-divergences were primarily investigated for probability measures,\footnote{In that context, a $d$-dimensional point is interpreted as a discrete and finite probability measure lying in the $(d-1)$-dimensional unit simplex.} we can extend the $f$-divergence to positive measures.
Since program (OPT) is  {\it strictly convex} in $x$, it admits a {\it unique} minimizer $M(\P;I_f)=\arg\min_{x} L(x; \P, I_f)$, termed the {\it entropic mean} by Ben-Tal et al.~\cite{entropicmeans-1989}.
Interestingly, those entropic means are linear scale-invariant:\footnote{That is, means of homogeneous degree $1$.} 

\begin{equation}
M(\lambda p_1, ..., \lambda p_n ; I_f)  = \lambda M(p_1, ..., p_n; I_f)  
\end{equation}
Nielsen and Nock~\cite{2009-BregmanCentroids-TIT} considered another class of information-theoretic distortion measures $B_F$ called Bregman divergences~\cite{CensorZenios-1997,WainwrightJordan-2008}:
\begin{equation}
B_F(x,p) = F(x)-F(p)-(x-p) F'(p),
\end{equation}
for a strictly convex differentiable function $F$.
It follows that (OPT) is convex, and admits a unique minimizer $M(p_1, ..., p_n; B_F)=M_{F'}(p_1, ..., p_n)$, a quasi-arithmetic mean for the strictly increasing and continuous function $F'$, the derivative of $F$. 
Observe that information-theoretic distances may be asymmetric (i.e., $d(x,p)\not = d(p,x)$), 
and therefore one may also define a {\it right-sided} centroid $M'$ as the minimizer of

\begin{equation}
\mathrm{(OPT')}: \min_{x} \sum_{i=1}^n w_i d(p_i,x),
\end{equation}

It turns out that for $f$-divergences, we have: 
\begin{equation}
I_f(x,p)=I_{f*}(p,x),
\end{equation} for $f^*(x)=x f(1/x)$ so that (OPT') is solved as a (OPT) problem for the {\it conjugate function} $f^*(\cdot)$.
In the same spirit, we have:
\begin{equation}
B_F(x,p)=B_{F^*}(F'(p),F'(x))
\end{equation}
for Bregman divergences, where $F^*$ denotes the Legendre convex conjugate~\cite{CensorZenios-1997,WainwrightJordan-2008}.\footnote{Legendre dual convex conjugates $F$ and $F^*$ have necessarily reciprocal gradients: ${F^*}'=(F')^{-1}$. See~\cite{2009-BregmanCentroids-TIT}.} 
Surprisingly, although (OPT') may {\it not} be convex in $x$ for Bregman divergences (e.g., $F(x)=-\log x$), (OPT') admits nevertheless a unique minimizer, independent of the generator function $F$: the center of mass $M'(\P;B_F) = \sum_{i=1}^n \frac{1}{n} p_i$.
Bregman means are not homogeneous except for the power generators $F(x)=x^p$ which yields entropic means, i.e. means that can {\it also} be interpreted\footnote{In fact, Amari~\cite{alphaunique-2009} proved that the intersection of the class of $f$-divergences with the class of Bregman divergences are $\alpha$-divergences.} as minimizers of average $f$-divergences~\cite{entropicmeans-1989}.
Amari~\cite{amari-2007} further studied those power means (known as $\alpha$-means in information geometry~\cite{informationgeometry-2000}), and showed that they are linear-scale free means obtained as minimizers of $\alpha$-divergences, a proper subclass of $f$-divergences.
Nielsen and Nock~\cite{isvd2009} reported an alternative simpler proof of $\alpha$-means by showing that the $\alpha$-divergences are Bregman divergences in disguise  (namely, representational Bregman divergences for positive measures, but not for normalized distribution measures~\cite{alphaunique-2009}). 
To get geometric centroids, we simply consider multivariate extensions of the optimization task (OPT).
In particular, one may consider {\it separable} divergences that are divergences that can be assembled coordinate-wise:
\begin{equation}
d(x,p)= \sum_{i=1}^d d_i(x^{(i)}, p^{(i)}),
\end{equation}
 with $x^{(i)}$ denoting the $i$th coordinate. 
A typical non separable divergence is the squared Mahalanobis distance~\cite{BregmanVoronoi-2010}:
\begin{equation}
d(x,p) = (x-p)^T Q (x-p),
\end{equation}
a Bregman divergence called generalized quadratic distance, defined for the generator $F(x)=x^T Q x$, where $Q$ is a positive-definite matrix ($Q\succ 0$).
For separable distances, the optimization problem (OPT) may then be reinterpreted as the task of finding the projection~\cite{meangenproj-1995} of a point $p$ (of dimension $d\times n$) to the upper line $U$:

\begin{equation}
\mathrm{(PROJ)} : \inf_{u\in U} d(u,p)
\end{equation}
with $u_1 =... = u_{d\times n} >0$, and $p$ the $(n\times d)$-dimensional point obtained by stacking the $d$ coordinates of each of the $n$ points.
\end{itemize}

In geometry, means (centroids) play a crucial role in center-based clustering (i.e., $k$-means~\cite{bregmankmeans-2005} for vector quantization applications).
Indeed, the mean of a cluster allows one to {\it aggregate data} into a single center datum. 
Thus the notion of means are encapsulated into the broader theory of mathematical aggregators~\cite{Detyniecki-mathematicalaggregation-2000}. 

Results on geometric means can be easily transfered to the field of Statistics~\cite{entropicmeans-1989} by generalizing the optimization problem task to a random variable $X$ with distribution $F$ as:

\begin{equation}
\mathrm{(OPT)}: \min_x E[X d(x,X)] = \min_x \int_t  t d(x,t) \dF(t),
\end{equation}
where $E[\cdot]$ denotes the expectation defined with respect to the Lebesgue-Stieltjes integral.
Although this approach is discussed in~\cite{entropicmeans-1989} and important for defining various notions of centrality in statistics, we shall not cover this extended framework here, for sake of brevity.

\subsection{Burbea-Rao divergences}
In this paper, we focus on the optimization approach (OPT) for defining other (geometric) means using the  class of information-theoretic distances obtained by Jensen difference for a strictly convex and differentiable function $F$:

\begin{equation}
d(x,p) = \frac{F(x)+F(p)}{2} - F\left( \frac{x+p}{2} \right) \equaldef\BR_F(x,p) \geq 0.
\end{equation}

Since the underlying differential geometry implied by those Jensen difference distances have been seminally studied in  papers of Burbea and Rao~\cite{BurbeaRao-1982,BurbeaRao-higherorder-1982}, we shall term them Burbea-Rao divergences, and point out to them as $\BR_F$.
In the remainder, we consider separable Burbea-Rao divergences.
That is, for $d$-dimensional points $p$ and $q$, we define
\begin{equation}
\BR_F(p,q) = \sum_{i=1}^d \BR_F(p^{(i)},q^{(i)}),
\end{equation}
and study the Burbea-Rao centroids (and barycenters) as the minimizers of the average Burbea-Rao divergences.
Those Burbea-Rao divergences generalize the celebrated Jensen-Shannon divergence~\cite{Jensen-Shannon-divergence} 
\begin{equation}
\JS(p,q)= H\left(\frac{p+q}{2}\right) - \frac{H(p)+H(q)}{2}
\end{equation}
by choosing $F(x)=-H(x)$, the negative Shannon entropy $H(x)=-x\log x$. 
Generators $F(\cdot)$ of parametric distances are convex functions representing entropies which are concave functions.
Burbea-Rao divergences contain all generalized quadratic distances ($F(x)=x^T Q x=\innerproduct{Qx}{x}$ for a positive definite matrix $Q\succ 0$, also called squared Mahalanobis distances):

\begin{eqnarray*}
\BR_F(p,q) & = & \frac{F(p)+F(q)}{2} - F\left(\frac{p+q}{2}\right)  \\
& = & \frac{2\innerproduct{Qp}{p}+ 2\innerproduct{Qq}{q}- \innerproduct{Q(p+q)}{p+q} }{4} \\
& = & \frac{1}{4} (\innerproduct{Qp}{p}+\innerproduct{Qq}{q}-2\innerproduct{Qp}{q}) \\
& = & \frac{1}{4} \innerproduct{Q(p-q)}{p-q} = \frac{1}{4} \| p-q \|^2_Q.
\end{eqnarray*}

Although the square root of the Jensen-Shannon divergence yields a metric (a Hilbertian metric), it is not true in general for Burbea-Rao divergences.
The closest work to our paper is a $1$-page symposium\footnote{In the nineties, the IEEE International Symposium on Information Theory (ISIT) published only $1$-page papers. We are grateful to Prof. Mich\`ele Basseville for sending us the corresponding slides.} paper~\cite{entdivmean-1995} discussing about  Ali-Silvey-Csisz\'ar $f$-divergences~\cite{AliSilvey-1966,Csiszar-1967} and Bregman divergences~\cite{Bregman67,CensorZenios-1997} (two entropy-based divergence classes).
Those information-theoretic distortion classes are compared using quadratic differential metrics, mean values and projections.
The notion of skew Jensen differences intervene in the discussion.

\subsection{Contributions and paper organization}
The paper is articulated into two parts: 
The first part studies the Burbea-Rao centroids, and the second part shows some applications in Statistics.
We summarize our contributions as follows:

\begin{itemize}
\item We define the parametric class of (skew) Burbea-Rao divergences, and show that those divergences naturally arise when generalizing the principle of the Jensen-Shannon divergence~\cite{Jensen-Shannon-divergence} to Jensen-Bregman divergences. In the limit cases, we further prove that those skew Burbea-Rao divergences yield asymptotically Bregman divergences.

\item We show that the centroids with respect to the (skew) Burbea-Rao divergences are unique.
Besides centroids for special cases of Burbea-Rao divergences (including the squared Euclidean distances), those centroids are not available in closed-form equations.
However, we show that any Burbea-Rao centroid can be estimated efficiently using an iterative convex-concave optimization procedure.  
As a by-product, we find Bregman sided centroids~\cite{2009-BregmanCentroids-TIT} in closed-form in the extremal skew cases.
\end{itemize}

We then consider applications of Burbea-Rao centroids in Statistics, and show the link with Bhattacharyya distances.
A wide class of statistical parametric models can be handled  in a unified manner as  exponential families~\cite{ef-flashcards-2009}.
The classes of exponential families contain many of the standard parametric models including the Poisson, Gaussian, multinomial, and Gamma/Beta distributions, just to name a few prominent members. 
However, only a few closed-form formulas for the statistical Bhattacharyya distances  between those densities are reported in the literature.\footnote{For instance, the Bhattacharyya distance between multivariate normal distributions is given here~\cite{Fukunaga90}.}

For the second part, our contributions are reviewed as follows:

\begin{itemize}
\item We show that the (skew) Bhattacharyya distances calculated for distributions belonging to the same exponential family in statistics, are  equivalent to (skew) Burbea-Rao divergences. We mention corresponding closed-form formula for computing Chernoff coefficients and $\alpha$-divergences of exponential families. In the limit case, we obtain an alternative proof showing that the Kullback-Leibler divergence of members of the same exponential family is equivalent to a Bregman divergence calculated on the natural parameters~\cite{BregmanVoronoi-2010}.

\item We approximate iteratively the Bhattacharyya centroid of any set of distributions of the same exponential family (including multivariate Gaussians) using the Burbea-Rao centroid algorithm. For the case of multivariate Gaussians, we design yet another tailored iterative scheme based on matrix differentials, generalizing the former univariate study of Rigazio et al.~\cite{BattGaussian}. Thus we get either the generic way or the tailored way for computing the Bhattacharrya centroids of arbitrary Gaussians.

\item As a field application, we show how to simplify Gaussian mixture models using hierarchical clustering, 
and show experimentally that the results obtained with the Bhattacharyya centroids compare favorably well with former results obtained for Bregman centroids~\cite{2010-Hierachical-ICCASP}. Our numerical experiments show that the generic method outperforms the alternative tailored method for multivariate Gaussians.
\end{itemize}

The paper is organized as follows: 
In section~\ref{sec:jb}, we introduce Burbea-Rao divergences as a natural extension of the Jensen-Shannon divergence using the framework of Bregman divergences. 
It is followed by Section~\ref{sec:sbr} which considers the general case of skew divergences, and reveals asymptotic behaviors of extreme skew Burbea-Rao divergences as Bregman divergences.
Section~\ref{sec:skewc} defines the (skew) Burbea-Rao centroids, show they are unique, and present a simple iterative algorithm with guaranteed convergence.
We then consider applications in Statistics in Section~\ref{sec:bhatbj}: After briefly recalling exponential distributions in \S\ref{sec:ef}, we show that Bhattacharyya distances and Chernoff/Amari $\alpha$-divergences are available in closed-form equations as Burbea-Rao divergences for distributions of the same exponential families. Section~\ref{sec:alt} presents an alternative iterative algorithm tailored to compute the Bhattacharyya centroid of multivariate Gaussians, generalizing the former specialized work of Rigazio et al.~\cite{BattGaussian}. 
In section~\ref{sec:clustering}, we use those Bhattacharyya/Burbea-Rao centroids to simplify hierarchically Gaussian mixture models, and comment both qualitatively and quantitatively our experiments on a color image segmentation application. 
Finally, section~\ref{sec:concl} concludes this paper by  describing  further perspectives and hinting at some information geometrical aspects of this work.

\section{Burbea-Rao divergences from symmetrization of Bregman divergences}\label{sec:jb}

Let $\mathbb{R}^+=[0,+\infty)$ denote the set of non-negative reals.
For a strictly convex (and differentiable) generator $F$, we define the Burbea-Rao divergence as the following non-negative function:

\begin{eqnarray*}
\BR_F &: & \mathcal{X}\times \mathcal{X} \rightarrow \mathbb{R}^+ \\
(p,q) &\mapsto & \BR_F(p,q)= \frac{F(p)+F(q)}{2} - F\left(\frac{p+q}{2}\right) \geq 0
\end{eqnarray*}
The non-negative property of those divergences follows straightforwardly from Jensen inequality.
Although Burbea-Rao distances are symmetric ($\BR_F(p,q)=\BR_F(q,p)$), they are not metrics since they fail to satisfy the triangle inequality.
A geometric interpretation of those divergences is given in Figure~\ref{fig:br}.
Note that $F$ is defined up to an affine term $ax+b$.

\begin{figure}
\centering
\includegraphics[width=\columnwidth]{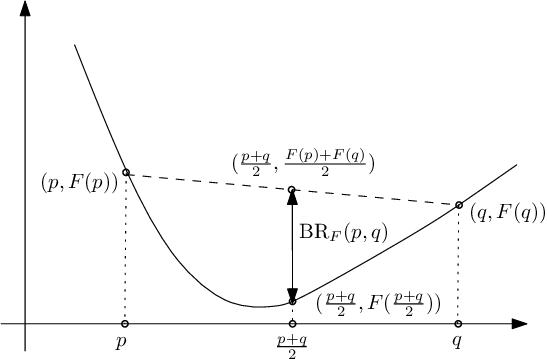}
\caption{Interpreting the Burbea-Rao divergence $\BR_F(p,q)$ as the vertical distance between the midpoint of segment $[(p,F(p)) , (q,F(q))]$ and the midpoint of the graph plot $\left(\frac{p+q}{2},F\left(\frac{p+q}{2}\right)\right)$. \label{fig:br} }
\end{figure}

We show that Burbea-Rao divergences extend the Jensen-Shannon divergence using the broader concept of Bregman divergences instead of the Kullback-Leibler divergence.
A Bregman divergence~\cite{Bregman67,CensorZenios-1997,WainwrightJordan-2008} $B_F$ is defined as the positive tail of the first-order Taylor expansion of a strictly convex and differentiable convex function $F$:

\begin{equation}
B_F(p,q)=F(p)-F(q)-\innerproduct{p-q}{\nabla F(q)},
\end{equation}
where $\nabla F$ denote the gradient of $F$ (the vector of partial derivatives $\{\frac{\partial F}{\partial x_i}\}_{i}$), and $\innerproduct{x}{y}=x^Ty$ the inner product (dot product for vectors).
A Bregman divergence is interpreted geometrically~\cite{BregmanVoronoi-2010} as the vertical distance between the tangent plane $H_q$ at $q$ of the graph plot $\mathcal{F}=\{ \hat x = (x,F(x))\ | x\in\mathcal{X}\}$ and its translates $H_q'$ passing through $\hat p=(p,F(p))$. 
Figure~\ref{fig:bd} depicts graphically the geometric interpretation of the Bregman divergence (to be compared with the Burbea-Rao divergence in Figure~\ref{fig:br}).

\begin{figure}
\centering
\includegraphics[width=\columnwidth]{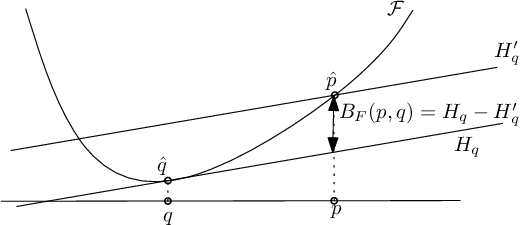}
\caption{Interpreting the Bregman divergence $B_F(p,q)$ as the vertical distance between the tangent plane at $q$ and its translate passing through $p$ (with identical slope $\nabla F(q)$).\label{fig:bd}}
\end{figure}

Bregman divergences are never metrics, and symmetric only for the generalized quadratic distances~\cite{BregmanVoronoi-2010} obtained by choosing $F(x)=x^T Q x$, for some positive definite matrix $Q\succ 0$.
Bregman divergences allow one to encapsulate both statistical distances with geometric distances: 

\begin{itemize} 

\item Kullback-Leibler divergence obtained for $F(x)=x\log x$:
\begin{equation}
\KL(p,q)= \sum_{i=1}^d p^{(i)}\log \frac{p^{(i)}}{q^{(i)}}
\end{equation}

\item squared Euclidean distance obtained for $F(x)=x^2$:
\begin{equation}
L_2^2(p,q)= \sum_{i=1}^d (p^{(i)}-q^{(i)})^2 = \| p-q\|^2
\end{equation}
\end{itemize}

Basically, there are two ways to symmetrize Bregman divergences (see also work on Bregman metrization~\cite{metricbregman1-2008,metricbregman2-2008}):

\begin{itemize}

\item {\bf Jeffreys-Bregman divergences.} We consider half of the double-sided divergences:

\begin{eqnarray}
S_F(p;q) &= & \frac{B_F(p,q)+B_F(q,p)}{2}\\
& =& \frac{1}{2}\innerproduct{p-q}{\nabla F(p)-\nabla F(q)},
\end{eqnarray} 

Except for the generalized quadratic distances, this symmetric distance {\it cannot} be interpreted as a Bregman divergence~\cite{BregmanVoronoi-2010}.

\item {\bf Jensen-Bregman divergences.} We consider the Jeffreys-Bregman divergences from the source parameters to the average parameter $\frac{p+q}{2}$ as follows:

\begin{eqnarray}
J_F(p;q) &= & \frac{ B_F(p,\frac{p+q}{2}) + B_F(q,\frac{p+q}{2}) }{2} \\
& = & \frac{F(p)+F(q)}{2}-F(\frac{p+q}{2})  = \BR_F(p,q)\nonumber
\end{eqnarray}

\end{itemize}

Note that even for the negative Shannon entropy $F(x)=x\log x-x$ (extended to positive measures), those two symmetrizations yield different divergences: While $S_F$ uses the gradient $\nabla F$, $J_F$ relies only on the generator $F$.
Both $J_F$ and $S_F$ have always finite values.\footnote{This may not be the case of Bregman/Kullback-Leibler divergences that can potentially be unbounded.}
The first symmetrization approach was historically studied by Jeffreys~\cite{Jeffreys46}.

The second way to symmetrize Bregman divergences generalizes the spirit of the Jensen-Shannon divergence~\cite{Jensen-Shannon-divergence}
\begin{eqnarray}
\JS(p,q) & = &\frac{1}{2} \left( \KL\left(p , \frac{p+q}{2}\right) + \KL\left(q , \frac{p+q}{2}\right) \right), \\
 & = & H\left(\frac{p+q}{2}\right) - \frac{H(p)+H(q)}{2}
\end{eqnarray}
 with non-negativity that can be derived from Jensen's inequality, hence its name.
The Jensen-Shannon divergence is also called the total divergence to the average, a generalized measure of diversity from the {\it population distributions} $p$ and $q$ to the {\it average population} $\frac{p+q}{2}$.
Those Jensen difference-type divergences are by definition Burbea-Rao divergences.
For the Shannon entropy, those two different information divergence symmetrizations (Jensen-Shannon divergence and Jeffreys $J$ divergence) satisfy the following inequality:
\begin{equation}
J(p,q) \geq 4\ \JS(p,q) \geq 0.
\end{equation}
Nielsen and Nock~\cite{2009-BregmanCentroids-TIT} investigated the centroids with respect to Jeffreys-Bregman divergences (the symmetrized Kullback-Leibler divergence).


\section{Skew Burbea-Rao divergences}\label{sec:sbr}

We further generalize Burbea-Rao divergences by introducing a positive weight $\alpha\in (0,1)$ when averaging source parameters $p$ and $q$ as follows:
\begin{eqnarray*}
\BR_F^{(\alpha)} &: & \mathcal{X}\times \mathcal{X} \rightarrow \mathbb{R}^+ \\
  \BR_F^{(\alpha)}(p,q) & = & \alpha F(p)+(1-\alpha) F(q) - F(\alpha p+(1-\alpha) q) 
\end{eqnarray*}
We consider the open interval $(0,1)$ since otherwise the divergence has no discriminatory power (indeed, for $\alpha\in\{0,1\},   \BR_F^{(\alpha)}(p,q)=0,\ \forall p,q$).
Although skewed divergences are asymmetric $\BR_F^{(\alpha)}(p,q)\not = \BR_F^{(\alpha)}(q,p)$, we can swap arguments by replacing $\alpha$ by $1-\alpha$:

\begin{eqnarray}
\BR_F^{(\alpha)}(p,q) & = & \alpha F(p) + (1-\alpha) F(q) - F(\alpha p+(1-\alpha) q) \nonumber\\
& = & \BR_F^{(1-\alpha)}(q,p) 
\end{eqnarray}
Those skew Burbea-Rao divergences are similarly found using a skew Jensen-Bregman counterpart (the gradient terms $\nabla F(\alpha p+(1-\alpha)q)$  perfectly cancel in the sum of skew Bregman divergences):

{
\begin{eqnarray*}
\alpha B_F(p,\alpha p+(1-\alpha)q)+(1-\alpha)B_F(q,\alpha p+(1-\alpha)q) \equaldef\\
 \BR_F^{(\alpha)}(p,q) 
\end{eqnarray*}
}

%
%

In the limit cases, $\alpha\rightarrow 0$ or $\alpha\rightarrow 1$, we have $\BR_F^{(\alpha)}(p,q)\rightarrow 0\ \forall p,q$.
That is, those divergences loose their discriminatory power at extremities.
However, we show that those skew Burbea-Rao divergences tend {\it asymptotically} to Bregman divergences:

\begin{eqnarray}
B_F(p,q) & = & \lim_{\alpha\rightarrow 0} \frac{1}{\alpha} \BR_F^{(\alpha)}(p,q) \label{eq:asbreg2}\\
B_F(q,p) & = & \lim_{\alpha\rightarrow 1} \frac{1}{1-\alpha} \BR_F^{(\alpha)}(p,q) \label{eq:asbreg}
\end{eqnarray}

The limit in the right-hand-side of Eq.~\ref{eq:asbreg} can be expressed alternatively as the following one-sided limit:

\begin{equation}\label{eq:lim}
\lim_{\alpha\uparrow 1} \frac{1}{1-\alpha} \BR_F^{(\alpha)}(p,q) = \lim_{\alpha\downarrow 0} \frac{1}{\alpha} \BR_F^{(\alpha)}(q,p),
\end{equation}
where the arrows $\uparrow$ and $\downarrow$ denote  the limit from the left and the limit from the right, respectively (see~\cite{LieseVajda-2006} for notations). The right derivative of a function $f$ at $x$ is defined as $f'_+(x)=\lim_{y\downarrow x}  \frac{f(y)-f(x)}{y-x}$.
Since $\BR_F^{(0)}(p,q)=0\ \forall p,q$, it follows that the right-hand-side limit of Eq.~\ref{eq:lim}  is the  right derivative (see Theorem~1 of~\cite{LieseVajda-2006} that gives a generalized Taylor expansion of convex functions) of the map 

\begin{equation}
L(\alpha): \alpha \mapsto \BR_F^{(\alpha)}(q,p)
\end{equation}
taken at $\alpha=0$.
Thus we have

\begin{equation}
\lim_{\alpha\downarrow 0} \frac{1}{\alpha} \BR_F^{(\alpha)}(q,p) = L'_+(0).,
\end{equation}
with

\def\dalpha{\mathrm{d}\alpha}

\begin{eqnarray}
L'_+(0) & = & \frac{\mathrm{d}_+}{\dalpha} (\alpha F(q)+(1-\alpha) F(p)-F(\alpha q+(1-\alpha)p))\nonumber\\
& = & F(q)-F(p)-\innerproduct{q-p}{\nabla F(p)} \\
& = & B_F(q,p)
\end{eqnarray}

\begin{lemma}
Skew Burbea-Rao divergences tend asymptotically to Bregman divergences ($\alpha\rightarrow 0$) or reverse Bregman divergences ($\alpha\rightarrow 1$).
\end{lemma}

Thus we may scale skew Burbea-Rao divergences so that Bregman divergences belong to skew Burbea-Rao divergences:

\begin{eqnarray}
\lefteqn{
\mathrm{sBR}_F^{(\alpha)}(p,q)  =}\nonumber\\
&& \frac{1}{\alpha(1-\alpha)} \left( \alpha F(p) + (1-\alpha) F(q) - F(\alpha p+(1-\alpha) q) \right)\nonumber\\
\ 
\end{eqnarray}

Moreover, $\alpha$ is now not anymore restricted to $(0,1)$ but to the full real line: $\alpha\in\mathbb{R}$, as also noticed in~\cite{zhang-2004}.
Setting $\alpha=\frac{1-\alpha'}{2}$ (that is, $\alpha'=1-2\alpha$), we get

\begin{eqnarray}
\lefteqn{
\mathrm{sBR}_F^{(\alpha')}(p,q)  =}\nonumber\\
&& \frac{4}{1-\alpha'^2} \left( \frac{1-\alpha'}{2} F(p) + \frac{1+\alpha'}{2} F(q) - F\left( \frac{1-\alpha'}{2}p+\frac{1+\alpha'}{2}q \right) \right)\nonumber\\
\ 
\end{eqnarray}

\section{Burbea-Rao centroids}\label{sec:skewc}

Let $\P=\{p_1, ..., p_n\}$ denote a $d$-dimensional point set.
To each point, let us further associate a positive weight $w_i$ (accounting for arbitrary multiplicity) and a positive scalar $\alpha_i\in (0,1)$ to define an anchored distance $\BR_F^{(\alpha_i)}(\cdot,p_i)$. 
Define the skew Burbea-Rao\footnote{We also call them skew Jensen barycenters or centroids since they are induced by a divergence using the Jensen inequality.} barycenter (or centroid) $c$ as the minimizer of the following optimization task:

\begin{equation} 
\mathrm{OPT}: c=\arg\min_x \sum_{i=1}^n w_i \BR_F^{(\alpha_i)}(x,p_i) = \arg\min_x L(x)
\end{equation}

Without loss of generality, we consider argument $x$ on the left argument position (otherwise, we change all $\alpha_i\rightarrow 1-\alpha_i$ to get the right-sided Burbea-Rao centroid).
Removing all terms independent of $x$, the minimization program (OPT) amounts to minimize equivalently the following energy function:

\begin{equation}
E(c) = (\sum_{i=1}^n {w_i\alpha_i})F(c)  - \sum_{i=1}^n w_i F(\alpha_i c+ (1-\alpha_i) p_i)
\end{equation}

Observe that the energy function is decomposable in the sum of a convex function $(\sum_{i=1}^n {w_i\alpha_i})F(c)$ with a concave function $- \sum_{i=1}^n w_i F(\alpha_i c+ (1-\alpha_i) p_i)$ (since the sum of $n$ concave functions is concave).
We can thus solve iteratively this optimization problem using the Convex-ConCave Procedure~\cite{YuilleNC03,SriperumbudurNIPS09} (CCCP), by starting from an initial position $c_0$ (say, the barycenter $c_0=\sum_{i=1}^n w_i p_i$), and iteratively update the barycenter as follows:

\begin{equation}
\nabla F(c_{t+1})  =  \frac{1}{ \sum_{i=1}^n {w_i\alpha_i} } \sum_{i=1}^n w_i\alpha_i \nabla F \left(\alpha_i c_t+ (1-\alpha_i) p_i\right) 
\end{equation}
\begin{equation}\label{eq:gm}
c_{t+1} =   \nabla F^{-1} \left(
\frac{1}{ \sum_{i=1}^n {w_i\alpha_i} } \sum_{i=1}^n w_i\alpha_i \nabla F \left(\alpha_i c_t+ (1-\alpha_i) p_i\right)
\right)
\end{equation}

Since $F$ is convex, the second-order derivative $\nabla ^2F$ is always positive definite, and $\nabla F$ is strictly monotone increasing.
Thus we can interpret Eq.~\ref{eq:gm} as a fixed-point equation by considering the $\nabla F$-representation.
Each iteration is interpreted as a quasi-arithmetic mean.
This proves that the Burbea-Rao centroid is always well-defined and unique (see Appendix~\ref{sec:proof} for a detailed proof), since there is (at most) a unique fixed point for $x=g(x)$ with a function $g(\cdot)$ strictly monotone increasing.

In some cases, like the squared Euclidean distance (or squared Mahalanobis distances), we find closed-form solutions for the Burbea-Rao barycenters.
For example, consider the (negative) quadratic entropy $F(x)=\innerproduct{x}{x} = \sum_{i=1}^d (x^{(i)})^2$ with weights $w_i$ and all $\alpha_i=\frac{1}{2}$ (non-skew  symmetric Burbea-Rao divergences).
We have:

\begin{eqnarray}\label{eq:solve}
& & \min E(x)=\frac{F(x)}{2}-\sum_{i=1}^n w_i F\left(\frac{p_i+x}{2}\right), \\
& = & \min \frac{\innerproduct{x}{x}}{2}-\frac{1}{4} \sum_{i=1}^n w_i\left(\innerproduct{x}{x} + 2  \innerproduct{x}{p_i} + \innerproduct{p_i}{p_i} \right)\nonumber
\end{eqnarray}

The minimum is obtained when the gradient $\nabla E(x)=0$, that is when $x=\bar p=\sum_{i=1}^n w_i p_i$, the barycenter of the point set $\mathcal{P}$.
For most Burbea-Rao divergences, Eq.~\ref{eq:solve} can only be solved numerically.

Observe that for extremal skew cases (for $\alpha\rightarrow 0$ or $\alpha\rightarrow 1$), we obtain the Bregman centroids in closed-form solutions (see Eq.~\ref{eq:asbreg}). 
Thus skew Burbea-Rao centroids allow one to get a smooth transition from the right-sided centroid (the center of mass) to the left-sided centroid (a quasi-arithmetic mean $M_f$ obtained for $f=\nabla F$, a continuous and strictly increasing function).

\begin{theorem}
Skew Burbea-Rao centroids are unique.
They can be estimated iteratively using the CCCP iterative algorithm.
In extremal skew cases,  the Burbea-Rao centroids tend to Bregman left/right sided centroids, and have closed-form equations in limit cases.
\end{theorem}

To describe the orbit of Burbea-Rao centroids linking the left to right sided Bregman centroids, we compute for $\alpha\in [0,1]$ the skew Burbea-Rao centroids with the following update scheme:

\begin{equation}
c_{t+1} = \nabla F^{-1} \left(
\sum_{i=1}^n w_i \nabla F(\alpha c_t + (1-\alpha)p_i)
 \right)
\end{equation}

We may further consider various convex generators $F_i$ for each point, and consider the updating scheme
{\small
\begin{eqnarray*}
\lefteqn{c_{t+1}= }\\ 
& &  \left(\sum_i w_i \nabla F_i\right)^{-1} \left( 
\frac{1}{ \sum_{i=1}^n {w_i\alpha_i} } \sum_{i=1}^n w_i\alpha_i \nabla F_i (\alpha_i c_t+ (1-\alpha_i) p_i)
\right)
\end{eqnarray*}
}

\subsection{Burbea-Rao divergences of a population}

Consider now the Burbea-Rao divergence of a population $p_1, ..., p_n$ with respective positive normalized weights $w_1, ..., w_n$.
The Burbea-Rao divergence is defined by:

\begin{equation}
\BR_F^w(p_1, ..., p_n) =  \sum_{i=1}^n w_i F(p_i) - F(\sum_{i=1}^n w_ip_i)  \geq 0
\end{equation}

This family of diversity measures includes the Jensen-R\'enyi divergences~\cite{jensenrenyi-isar-2001,JensenRenyi-CSPL328-2001} for $F(x)=-R_{\alpha}(x)$, where $R_\alpha(x)=\frac{1}{1-\alpha}\log \sum_{j=1}^d p_j^{\alpha}$ is the R\'enyi entropy of order $\alpha$.
(R\'enyi entropy is concave for $\alpha\in(0,1)$ and tend to Shannon entropy for $\alpha\to 1$.) 


\section{Bhattacharyya distances as Burbea-Rao distances}\label{sec:bhatbj}

We first briefly recall the versatile class of exponential family distributions in Section~\ref{sec:ef}.
Then we show in Section~\ref{sec:chernoff} that the statistical Bhattacharyya/Chernoff distances between exponential family distributions amount to compute a Burbea-Rao divergence.

\subsection{Exponential family distribution in Statistics}\label{sec:ef}

Many usual statistical parametric distributions $p(x;\lambda)$ (e.g., Gaussian, Poisson, Bernoulli/multinomial, Gamma/Beta, etc.) share common properties arising from their common
canonical decomposition of probability distribution~\cite{WainwrightJordan-2008}:

\begin{equation}
p(x;\lambda)=p_F(x;\theta)=\exp \left(\innerproduct{t(x)}{\theta}-F(\theta)+k(x)\right). \label{eq:ef}
\end{equation}

Those distributions\footnote{The distributions can either be discrete or continuous. We do not introduce the unifying framework of probability measures in order to not burden the paper.} are said to belong to the exponential families (see~\cite{ef-flashcards-2009} for a tutorial).
An exponential family is characterized by its {\it log-normalizer} $F(\theta)$, and a distribution in  that family by its {\it natural parameter} $\theta$ belonging to the {\it natural space} $\Theta$.
The log-normalizer $F$ is strictly convex and  $C^\infty$, and can also be expressed using the source coordinate system $\lambda$ using the 1-to-1 map $\tau: \Lambda \rightarrow \Theta$ that converts parameters from the source coordinate system $\lambda$ to the natural coordinate system $\theta$:

\begin{equation}
F(\theta)= F(\tau(\lambda)) = (F\circ\tau) (\lambda)= F_\lambda(\lambda),
\end{equation}
where $F_\lambda=F\circ\tau$ denotes the log-normalizer function expressed using the $\lambda$-coordinates instead of the natural $\theta$-coordinates.

The vector $t(x)$ denote the {\it sufficient statistics}, that is the set of linear independent functions that allows to concentrate without any loss all information about the parameter $\theta$ carried in the iid. observations $x_1, x_2, ...,$ .  
The inner product $\innerproduct{p}{q}$ is defined according to the primitive type of $\theta$. 
Namely, it is a multiplication $\innerproduct{p}{q}=pq$ for scalars, a dot product $\innerproduct{p}{q}=p^T q$ for vectors, a matrix trace  
$\innerproduct{p}{q}=\tr{p^T\times q}=\tr{p\times q^T}$ for matrices, etc.
For composite types such as $p$ being defined by both a vector part and a matrix part, the composite inner product is defined as the sum of inner products on the primitive types.
Finally, $k(x)$ represents the carrier measure according to the counting or Lebesgue measures. 
Decompositions for most common exponential family distributions are given in~\cite{ef-flashcards-2009}.
An exponential family $\E_F=\{p_F(x;\theta) \ | \theta\in\Theta\}$ is the set of probability distributions obtained for the same log-normalizer function $F$. Information geometry considers $\E_F$ as a manifold entity, and study its differential geometric properties~\cite{informationgeometry-2000}.

For example, consider the family of Poisson distributions $\E_F$ with mass function:

\begin{equation}
p(x;\lambda)=\frac{\lambda^x}{x!} \exp(-\lambda),
\end{equation}
for $x\in\mathbb{N}_+=\mathbb{N}\cup\{0\}$ a positive integer. 
Poisson distributions are univariate exponential families ($x\in\mathbb{N}_+$) of order $1$ (parameter $\lambda$).
The canonical decomposition yields

\begin{itemize}
\item the sufficient statistic $t(x)=x$,
\item $\theta=\log\lambda$, the natural parameter,
\item $F(\theta)=\exp\theta$, the log-normalizer,
\item and $k(x)=-\log x!$ the carrier measure (with respect to the counting measure).
\end{itemize}

Since we deal with applications using multivariate normals in the following, we also report explicitly that canonical decomposition for the multivariate Gaussian family $\{p_F(x;\theta)\ | \theta\in\Theta \}$. We rewrite the usual Gaussian density of mean $\mu$ and variance-covariance matrix $\Sigma$:

\begin{eqnarray}
p(x;\lambda) &= & p(x;\mu,\Sigma)\\
& =& \frac{1}{2\pi\sqrt{\det\Sigma}} \exp \left( - \frac{(x-\mu)^T\Sigma^{-1} (x-\mu))}{2} \right)
\end{eqnarray}

in the canonical form of Eq.~\ref{eq:ef} with,
\begin{itemize}
\item $\theta=(\Sigma^{-1}\mu, \frac{1}{2}\Sigma^{-1})\in \Theta=\mathbb{R}^d\times\mathbb{K}_{d\times d}$, with $\mathbb{K}_{d\times d}$ denotes the cone of positive definite matrices,

\item $F(\theta)=\frac{1}{4}\tr{\theta_2^{-1} \theta_1\theta_1^T}-\frac{1}{2}\log\det \theta_2+\frac{d}{2}\log\pi$,

\item $t(x)=(x,-x^Tx)$,

\item $k(x)=0$.
\end{itemize}

In this case, the inner product is composite and is calculated as the sum of a dot product and a matrix trace as follows:
\begin{equation}
\innerproduct{\theta}{\theta'}=\theta_1^T\theta_1'+\tr{\theta_2^T \theta_2'}.
\end{equation}
The coordinate transformation $\tau: \Lambda \rightarrow \Theta$ is  given for $\lambda=(\mu,\Sigma)$ by

\begin{equation}
\tau(\lambda)= \left(\lambda_2^{-1}\lambda_1, \frac{1}{2}\lambda_2^{-1}\right),
\end{equation}
and its inverse mapping $\tau^{-1}: \Theta \rightarrow \Lambda $  by

\begin{equation}
\tau^{-1}(\theta) = \left(\frac{1}{2}\theta_2^{-1}\theta_1, \frac{1}{2}\theta_2^{-1} \right).
\end{equation}

\subsection{Bhattacharyya/Chernoff coefficients and $\alpha$-divergences as skew Burbea-Rao divergences}\label{sec:chernoff}
For arbitrary probability distributions $p(x)$ and $q(x)$ (parametric or not), we measure the amount of overlap between those distributions using the Bhattacharyya coefficient~\cite{Bhatta1943}:

\begin{equation}
C_(p,q)=\int \sqrt{p(x)q(x)} \dx,
\end{equation}
Clearly, the Bhattacharyya coefficient (measuring the affinity between distributions~\cite{Matusita-1955}) falls in the unit range:
\begin{equation}
0\leq C(p,q)\leq 1.
\end{equation}
In fact, we may interpret this coefficient geometrically by considering $\sqrt{p(x)}$ and $\sqrt{q(x)}$ as unit vectors.
The Bhattacharyya distance is then the dot product, representing the cosine of the angle made by the two unit vectors.
The Bhattacharyya distance $B: \mathcal{X}\times\mathcal{X} \rightarrow \mathbb{R}^+$ is derived  from its
coefficient~\cite{Bhatta1943} as 

\begin{equation}
B(p,q)=-\ln C(p,q).
\end{equation}

The Bhattacharyya distance allows one to get {\it both} upper and lower bound the Bayes' classification error~\cite{KailathITCT67,Bhat1998}, while there are no such results for  the symmetric Kullback-Leibler divergence.
Both the Bhattacharyya distance and the symmetric Kullback-Leibler divergence agrees with the Fisher information at the infinitesimal level.
Although the Bhattacharyya distance is symmetric, it is not a metric.
Nevertheless, it can be metrized by transforming it into to the following Hellinger metric~\cite{hellinger-1907}:

\begin{equation}
H(p,q)=
\sqrt{\frac{1}{2} \int (\sqrt{p(x)}-\sqrt{q(x)})^2 \dx} ,
\end{equation}

such that $0\leq H(p,q)\leq 1$.
It follows that

\begin{eqnarray}
\lefteqn{H(p,q) =} \nonumber\\
& &  \sqrt{\frac{1}{2} \left(\int p(x)\dx+\int q(x)\dx -2\int \sqrt{p(x)} \sqrt{q(x)} \dx\right)}\nonumber\\
& = & \sqrt{1-C(p,q)}.
\end{eqnarray}

Hellinger metric is also called Matusita metric~\cite{Matusita-1955} in the literature.
The thesis of Hellinger was emphasized in the work of Kakutani~\cite{kakutani-1948}.

We consider a direct generalization of Bhattacharyya coefficients and divergences called Chernoff divergences\footnote{In the literature, Chernoff information is also defined as  $-\log \inf_{\alpha\in [0,1]}\int p^{\alpha}(x) q^{1-\alpha}(x) \dx $. Similarly, Chernoff coefficients $C_\alpha(p,q)$ are defined as the supremum: 
$C_\alpha(p,q)=\sup_{\alpha\in [0,1]} \int p^{\alpha}(x)q^{1-\alpha}(x)\dx$.
 }

\begin{eqnarray}
B_\alpha(p,q) & = & -\ln \int_x p^{\alpha}(x) q^{1-\alpha}(x) \dx = -\ln C_\alpha(p,q)\\
 & = & -\ln \int_x q(x) \left(\frac{p(x)}{q(x)}\right)^{\alpha} \dx \\
  & = & -\ln E_q[ L^\alpha(x)]
\end{eqnarray}
 defined for some  $\alpha\in(0,1)$ (the Bhattacharyya divergence is obtained  for $\alpha=\frac{1}{2}$), where $E[\cdot]$ denote the expectation, and $L(x)=\frac{p(x)}{q(x)}$ the likelihood ratio.
 The term $\int_x p^{\alpha}(x) q^{1-\alpha}(x) \dx$ is called the Chernoff coefficient.
 The Bhattacharyya/Chernoff distance of members of the same exponential family yields a weighted asymmetric Burbea-Rao divergence (namely, a {\it skew} Burbea-Rao divergence):

\begin{equation}
B_\alpha(p_F(x;\theta_p),p_F(x;\theta_q))  =  \BR_F^{(\alpha)}(\theta_p,\theta_q)
\end{equation}
with
\begin{equation}
\BR_F^{(\alpha)}(\theta_p,\theta_q) =  \alpha F(\theta_p)+ (1-\alpha) F(\theta_q) - F(\alpha\theta_p+(1-\alpha)\theta_q)
\end{equation}
 
Chernoff coefficients are also related to $\alpha$-divergences, the canonical divergences in $\alpha$-flat spaces in information geometry~\cite{informationgeometry-2000} (p. 57):

{\small

\begin{equation}
D_\alpha(p||q) = \left \{
\begin{array}{lr}
 \frac{4}{1-\alpha^2} \left( 1 - \int p(x)^{\frac{1-\alpha}{2}} q(x)^{\frac{1+\alpha}{2}} \dx \right), &  \alpha\not = \pm 1,\\
 \int p(x)\log \frac{p(x)}{q(x)} \dx= \KL(p,q) , & \alpha=-1,\\
 \int q(x)\log \frac{q(x)}{p(x)} \dx= \KL(q,p) , & \alpha=1,\\
 \end{array}
 \right .
\end{equation}
}

The class of $\alpha$-divergences satisfy the following reference duality: $D_\alpha(p||q)=D_{-\alpha}(q||p)$.
Remapping $\alpha'=\frac{1-\alpha}{2}$ ($\alpha=1-2\alpha'$), we transform Amari $\alpha$-divergences to Chernoff $\alpha'$-divergences:\footnote{
Chernoff coefficients are also related to R\'enyi $\alpha$-divergence generalizing the Kullback-Leibler divergence:
$R_{\alpha}(p||q) =  \frac{1}{\alpha-1} \log \int_x p(x)^\alpha q^{1-\alpha}(x) \dx$ built on 
 R\'enyi  entropy $H_R^\alpha(p)  =  \frac{1}{1-\alpha}\log (\int_x p^\alpha(x)\dx -1)$.
The Tsallis entropy $H_T^\alpha(p)=\frac{1}{\alpha-1}(1-\int p(x)^\alpha \dx)$
can also be obtained from the R\'enyi entropy (and vice-versa) via the mappings:
$H_T^\alpha(p)  =  \frac{1}{1-\alpha} (e^{(1-\alpha)H_R^\alpha(p)}-1)$ and
$H_R^\alpha(p) =  \frac{1}{1-\alpha} \log (1+(1-\alpha) H_T^\alpha(p))$.
}

{\small
\begin{equation}
D_{\alpha'}(p,q) = \left \{
\begin{array}{lr}
 \frac{1}{\alpha' (1-\alpha')} \left( 1 - \int p(x)^{\alpha'} q(x)^{1-\alpha'} \dx \right), &  \alpha'\not\in\{0,1\},\\
 \int p(x)\log \frac{p(x)}{q(x)} \dx =\KL(p,q) , & \alpha'=1,\\
 \int q(x)\log \frac{q(x)}{p(x)} \dx =\KL(q,p) , & \alpha'=0,\\
 \end{array}
 \right .
\end{equation}
}

\begin{theorem}
The Chernoff $\alpha'$-divergence ($\alpha\not =\pm 1$) of distributions belonging to the same exponential family is given in closed-form by means of a  skewed Burbea-Rao divergence as:  $D_{\alpha'}(p,q)= \frac{1}{\alpha' (1-\alpha')} (1-e^{-\BR_F^{\alpha'}(\theta_p,\theta_q)})$,
with  $\BR_F^{(\alpha)}(\theta_p,\theta_q)= (\alpha F(\theta_p)-(1-\alpha)F(\theta_q)) - F(\alpha\theta_p-(1-\alpha)\theta_q)$.
Amari $\alpha$-divergence for members of the same exponential families amount to compute
$D_{\alpha}(p,q)= \frac{4}{1-\alpha^2} (1-e^{-\BR_F^{\left(\frac{1-\alpha}{2}\right)}(\theta_p,\theta_q)})$
\end{theorem}

\begin{figure*}[!t]
\normalsize
Let us compute the Chernoff coefficient for distributions belonging to the same exponential families.
Without loss of generality, let us   consider the reduced canonical form of exponential families 
$p_F(x;\theta)=\exp \innerproduct{x}{\theta}-F(\theta)$.
Chernoff coefficients $C_\alpha(p,q)$ of members $p=p_F(x;\theta_p)$  and $q=p_F(x;\theta_q)$ of the {\it same} exponential family $\E_F$:
\begin{eqnarray*}
C_\alpha(p,q) &= &  \int p^{\alpha}(x) q^{1-\alpha}(x) \dx = \int p_F^{(\alpha)}(x;\theta_p) p_F^{1-\alpha}(x;\theta_q) \dx\\
 & = & \int \exp (\alpha(\innerproduct{x}{\theta_p}-F(\theta_p))) \times \exp ((1-\alpha)(\innerproduct{x}{\theta_q}-F(\theta_q))) \dx\\
 & = & \int \exp \left( \Innerproduct{x}{\alpha\theta_p+(1-\alpha)\theta_q}-(\alpha F(\theta_p)+(1-\alpha)F(\theta_q) \right)\dx\\
 & = & \exp -(\alpha F(\theta_p)+(1-\alpha)F(\theta_q)) \times \int \exp \left( \innerproduct{x}{\alpha\theta_p+(1-\alpha)\theta_q}-F(\alpha\theta_p+(1-\alpha)\theta_q)+F(\alpha\theta_p+(1-\alpha)\theta_q) \right) \dx\\
 & = & \exp \left( F(\alpha\theta_p+(1-\alpha)\theta_q) - (\alpha F(\theta_p)+(1-\alpha)F(\theta_q) \right) \times \int \exp \innerproduct{x}{\alpha\theta_p+(1-\alpha)\theta_q}-F(\alpha\theta_p+(1-\alpha)\theta_q)\dx \\
 & = & \exp \left( F(\alpha\theta_p+(1-\alpha)\theta_q) - (\alpha F(\theta_p)+(1-\alpha)F(\theta_q)\right) \times \underbrace{\int p_F(x;\alpha\theta_p+(1-\alpha)\theta_q)\dx}_{=1}\\
 & = & \exp (-\BR_F^{(\alpha)}(\theta_p,\theta_q)) \geq 0.
\end{eqnarray*}
\hrulefill
\vspace*{4pt}
\end{figure*}

We get the following theorem for Bhattacharyya/Chernoff distances:
\begin{theorem}
The skew Bhattacharyya divergence $B_\alpha(p,q)$ is equivalent to the Burbea-Rao divergence for members of the same exponential family $\E_F$:
$
B_\alpha(p,q) = B_\alpha(p_F(x;\theta_p),p_F(x;\theta_q)) = -\log C_\alpha(p_F(x;\theta_p),p_F(x;\theta_q)) =  \BR_F^{(\alpha)}(\theta_p,\theta_q) \geq 0
$.
\end{theorem}

In particular, for $\alpha=\pm 1$, the Kullback-Leibler divergence of those exponential family distributions amount to compute a 
{\it Bregman divergence}~\cite{BregmanVoronoi-2010} (by taking the limit as $\alpha\rightarrow 1$ or $\alpha\rightarrow 0$).

\begin{corollary}
In the limit case $\alpha'\in\{0,1\}$, the $\alpha'$-divergences amount to compute a Kullback-Leibler divergence, and is equivalent to compute a Bregman divergence for the log-normalized on the swapped natural parameters: $\KL(p_F(x;\theta_p),p_F(x;\theta_q))=B_F(\theta_q,\theta_p)$.
\end{corollary}

\begin{proof}
The proof relies on the equivalence of Burbea-Rao divergences to Bregman divergences for extremal  values of $\alpha\in\{0,1\}$.

\begin{eqnarray}
\KL(p,q) &= & \KL(p_F(x;\theta_p), p_F(x;\theta_q))\\
  & =  &  \lim_{\alpha'\rightarrow 1} D_{\alpha'}(p_F(x;\theta_p), p_F(x;\theta_q)) \\
& = & \lim_{\alpha'\rightarrow 1} \frac{1}{\alpha'(1-\alpha')} (1-\underbrace{C_\alpha(p_F(x;\theta_p), p_F(x;\theta_q))}_{\mbox{since\ }\exp x \simeq_{x\simeq 0} 1+x}) \nonumber\\
& = & \lim_{\alpha'\rightarrow 1} \frac{1}{\alpha'(1-\alpha')} \underbrace{\BR_F^{\alpha'}(\theta_p,\theta_q)}_{(1-\alpha')B_F(\theta_q,\theta_p)} \\
& = & \lim_{\alpha'\rightarrow 1}  \frac{1}{\alpha'} B_F(\theta_q,\theta_p) = B_F(\theta_q,\theta_p) 
\end{eqnarray}

Similarly, we have $\lim_{\alpha'\rightarrow 0} D_{\alpha'}(p_F(x;\theta_p), p_F(x;\theta_q))=\KL(p_F(x;\theta_q), p_F(x;\theta_p))= B_F(\theta_p,\theta_q)$.
\end{proof}

Table~\ref{tab:exp:bhat} reports the Bhattacharyya distances for members of the same exponential families.

\begin{table*}
\centering
\renewcommand{\arraystretch}{1.5}
\begin{tabular}{llll}
Exponential family & 
$\tau : \lambda\rightarrow \theta$ & $F(\theta)$ (up to a constant) & Bhattacharyya/Burbea-Rao  $\BR_F(\lambda_p,\lambda_q)=\BR_F(\tau(\lambda_p), \tau(\lambda_q))$
 \\ \hline \hline
Multinomial & 
$(\log \frac{p_i}{p_d})_i$
&
$\log (1+\sum_{i=1}^{d-1} \exp\theta_i) $
&
$ -\ln \sum_{i=1}^d \sqrt{p_iq_i}$ \\
Poisson & 
$\log\lambda$
&
$
\exp \theta
$
&
$\frac{1}{2}(\sqrt{\mu_p}-\sqrt{\mu_q})^2$ \\
Gaussian &
$(\theta_1=\mu,\theta_2=\sigma^2)$
&
$-\frac{\theta_1^2}{4\theta_2}+\frac{1}{2}\log (-\frac{\pi}{\theta_2})$
&
 $\frac{1}{4} \frac{(\mu_p-\mu_q)^2}{\sigma_p^2+\sigma_q^2} +
\frac{1}{2} \ln \frac{\sigma_p^2+\sigma_q^2}{2\sigma_p\sigma_q}  $ \\
Multivariate Gaussian &
$(\theta=\Sigma^{-1}\mu,\Theta=\frac{1}{2}\Sigma^{-1})$
&
$
\frac{1}{4}\tr{\Theta^{-1}\theta\theta^T}-\frac{1}{2}\log\det\Theta 
$
&
 $\frac{1}{8} (\mu_p-\mu_q)^T
\left(\frac{\Sigma_p+\Sigma_q}{2} \right)^{-1} (\mu_p-\mu_q) +
\frac{1}{2}\ln \frac{\det \frac{\Sigma_p+\Sigma_q}{2}}{\det\Sigma_p
\det\Sigma_q} $
\end{tabular}

\caption{Closed-form Bhattacharyya distances for some classes of exponential families (expressed in source parameters for ease of use)).}\label{tab:exp:bhat}
\end{table*}


\subsection{Direct method for calculating the Bhattacharyya centroids of multivariate normals}\label{sec:alt}

To the best of our knowledge, the Bhattacharyya centroid has only  been studied for univariate Gaussian or diagonal multivariate Gaussian distributions~\cite{RigazioICASSP00} in the context of speech recognition, where it is reported that it can be estimated  using an iterative algorithm (no convergence guarantees are reported in~\cite{RigazioICASSP00}).

In order to compare this scheme on multivariate data with our generic Burbea-Rao scheme,
we extend the approach of Rigazio et al.~\cite{RigazioICASSP00} to multivariate Gaussians.
Plugging the Bhattacharyya distance of Gaussians in the energy function of the optimization problem (OPT), we get
\begin{eqnarray}
L(c)&=&\sum_{i=1}^{n} \frac{1}{8} \left( \mu_c - \mu_i\right)^T \left(\frac{\Sigma_c+\Sigma_i}{2}\right)^{-1} ( \mu_c - \mu_i )\nonumber\\ &+&\  \frac{1}{2}\log \left( { \det \left( \frac{\Sigma_c+\Sigma_i}{2} \right) \over \sqrt{\det \Sigma_c \det \Sigma_i}}  \right).  
\end{eqnarray}

This  is equivalent to minimize the following energy:
\begin{eqnarray}
F(c)&=&\sum_{i=1}^{n}  \left( \mu_c - \mu_i\right)^T \left(\Sigma_c+\Sigma_i\right)^{-1} ( \mu_c - \mu_i )\  \nonumber\\
&+&\  2\log \left( \det (\Sigma_c+\Sigma_i) \right) - \log \left(  \det  \Sigma_c   \right) \nonumber\\ 
&-& \log \left( 2^{2d} \det  \Sigma_i  \right).  
\end{eqnarray}
In order to minimize $F(c)$, let us differentiate with respect to $\mu_c$.
 let $U_i$ denote $\left(\Sigma_c+\Sigma_i\right)^{-1}$.
Using matrix differentials~\cite{Petersen04} (p.10 Eq. 73), we get: 
\begin{equation}
\frac{\partial L}{\partial \mu_c} = \sum_{i=1}^{n} \left[ U_i+U_i^T\right] [\mu_c - \mu_i]   
\end{equation}

Then one can estimate iteratively $\mu_c$, since $U_i$ depends on  $\Sigma_c$ which is unknown.
We update $\mu_c$ as follows:
\begin{equation}
\mu_c(t+1)=\left[\sum_{i=1}^{n} \left[ U_i+U_i^T\right]\right]^{-1}\left[\sum_{i=1}^{n} \left[ U_i+U_i^T\right]\mu_i\right]
\label{eq:bha:mu}
\end{equation}
Now let us estimate $\Sigma_c$. 
We used matrix differentials~\cite{Petersen04} (p.9 Eq. 55 for the first term, and  Eq. 51 p.8  for the two others): 
\begin{eqnarray}
\frac{\partial L}{\partial \Sigma_c}&=&\sum_{i=1}^{n} -U_i^T \left( \mu_c - \mu_i\right) \left( \mu_c - \mu_i\right)^T U_i^T \nonumber \\
&+& 2\sum_{i=1}^{n}  U_i^T - \sum_{i=1}^{n} \Sigma_c^{-T}. 
\end{eqnarray}
Taken into account the fact that $\Sigma_c$ is symmetric, differential calculus on symmetric matrices can be simply estimate:
\begin{equation}
\frac{d L}{d \Sigma_c} = \frac{\partial L}{\partial \Sigma_c}  + \left[\frac{\partial L}{\partial \Sigma_c}\right]^T -  \mathrm{diag}\left(\frac{\partial L}{\partial \Sigma_c}\right).
\end{equation}
Thus, if one notes 
\begin{equation}
A=  \sum_{i=1}^{n} 2 U_i ^T  -U_i^T \left( \mu_c - \mu_i\right) \left( \mu_c - \mu_i\right)^T U_i^T
\end{equation}
and recalling that $\Sigma_c$ is symmetric, one has to solve 
\begin{equation}
n (2 \Sigma_c^{-1}  - \mathrm{diag}(\Sigma_c^{-1})) = A + A^T - \mathrm{diag}(A).
\end{equation}
Let 
\begin{equation}
B = A + A^T - \mathrm{diag}(A)
\end{equation}
Then one can estimate $\Sigma_c$ iteratively as follows:
\begin{equation}
\Sigma_c^{(k+1)}= 2n \left[ ( B^{(k)} +  \mathrm{diag}(B^{(k)}))  \right ]^{-1}
\label{eq:bha:cov}
\end{equation}

Let us now compare the two generic Burbea-Rao/tailored Gaussian methods for computing the Bhattacharyya centroids on multvariate Gaussians.

\subsection{Applications to mixture simplification in statistics}\label{sec:clustering}

Simplifying Gaussian mixtures is important in many applications arising in signal processing~\cite{2010-Hierachical-ICCASP}.
Mixture simplification is also a crucial step when one wants to study the Riemannian geometry induced by the Rao distance with respect to the Fisher metric: The set of mixture models need to have the same number of components, so that we simplify source mixtures to get a set of Gaussian mixtures with prescribed size.
We adapt the hierarchical clustering algorithm of Garcia et al.~\cite{2010-Hierachical-ICCASP} by replacing the symmetrized Bregman centroid (namely, the Jeffreys-Bregman centroid) by the 
Bhattacharyya centroid.
We consider the task of color image segmentation by learning a Gaussian mixture model for each image.
Each image is represented as a set of $5D$ points (color $RGB$ and position $xy$).

The first experimental results depicted in Figure~\ref{fig:res:bhc} demonstrates the {\it qualitative stability} of the clustering performance. 
In particular, the  hierarchical clustering with respect to the Bhattacharrya distance performs qualitatively much better on the last {\tt colormap} image.\footnote{See reference images and segmentation using Bregman centroids at \url{http://www.informationgeometry.org/MEF/}}

\begin{figure*}
\centering
\def\wfig{0.20\linewidth}
\begin{tabular}{lcccc}
(a) & \includegraphics[bb=0 0 120 120, width=\wfig]{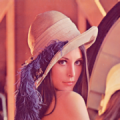}
&\includegraphics[bb=0 0 120 120,width=\wfig]{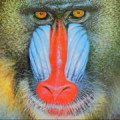}
&\includegraphics[bb=0 0 120 120,width=\wfig]{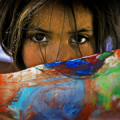}
&\includegraphics[bb=0 0 120 120,width=\wfig]{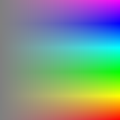}\vspace{0.3pt}\\
(b) &\includegraphics[bb=0 0 120 120,width=\wfig]{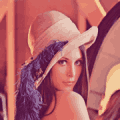}
&\includegraphics[bb=0 0 120 120,width=\wfig]{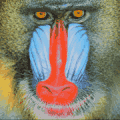}
&\includegraphics[bb=0 0 120 120,width=\wfig]{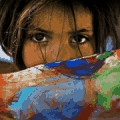}
&\includegraphics[bb=0 0 120 120,width=\wfig]{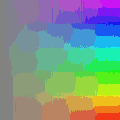}\vspace{0.3pt}\\
(c)&\includegraphics[bb=0 0 120 120,width=\wfig]{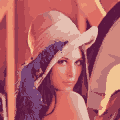}
&\includegraphics[bb=0 0 120 120,width=\wfig]{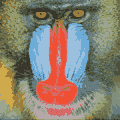}
&\includegraphics[bb=0 0 120 120,width=\wfig]{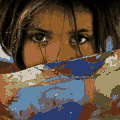}
&\includegraphics[bb=0 0 120 120,width=\wfig]{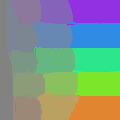}
\end{tabular}
\caption{Color image segmentation results:  (a) source images, (b) segmentation with $k=48$ 5D Gaussians, and (c) segmentation with $k=16$  5D Gaussians.}\label{fig:res:bhc}
\end{figure*}

The second experiment focuses on characterizing the numerical convergence of the generic Burbea-Rao method compared to the tailored Gaussian method.
Since we presented two novel different schemes to compute the Bhattacharyya centroids of multivariate Gaussians, one wants to compare them, both in terms of stability and accuracy. 
Whenever the ratio of Bhattacharyya distance energy function between those estimated centroids is greater than $1\%$, we  consider that one of the two  estimation methods is beaten (namely, the method that gives the highest Bhattacharyya distance).
Among the $760$ centroids computed to generate  Figures~\ref{fig:res:bhc}, $100\%$ were correct with the Burbea-Rao approach, while
only $87\%$ were correct with the tailored multivariate Gaussian matrix optimization method. 
The average number of iterations to reach the $1\%$ accuracy is $4.1$ for  the Burbea-Rao estimation algorithm, and $5.2$ for the alternative method. 

Thus we experimentally checked that the generic CCCP iterative Burbea-Rao algorithm described for computing the Bhattacharrya centroids always converge, and moreover beats  another {\it ad-hoc} iterative method tailored for multivariate Gaussians.

\section{Concluding remarks}\label{sec:concl}
In this paper, we have shown that the Bhattacharrya distance for distributions of the same statistical exponential families can be computed equivalently as a Burbea-Rao divergence on the corresponding natural parameters. 
Those results extend to skew Chernoff coefficients (and Amari $\alpha$-divergences) and skew Bhattacharyya distances using the notion of skew Burbea-Rao divergences.
We proved that (skew) Burbea-Rao centroids are unique, and can be efficiently estimated using an iterative concave-convex procedure with guaranteed convergence. We have shown that extremally skewed Burbea-Rao divergences amount asymptotically to evaluate Bregman divergences.
This work emphasizes on the attractiveness of exponential families in Statistics.
Indeed,  it turns out that for many statistical distances, one can  evaluate them in closed-form.
For sake of brevity, we have not mentioned the recent $\beta$-divergences and $\gamma$-divergences~\cite{alphabetagamma-2010}, although their distances on exponential families are again available in closed-form.

The differential Riemannian geometry induced by the class of such Jensen difference measures was studied by Burbea and Rao~\cite{BurbeaRao-1982,BurbeaRao-higherorder-1982} who built quadratic differential metrics on probability spaces using Jensen differences. 
The Jensen-Shannon divergence is also an instance of a broad class of divergences called the  $f$-divergences.
A $f$-divergence $I_f$ is a statistical measure of dissimilarity defined by the functional $I_f(p,q)=\int p(x) f(\frac{q(x)}{p(x)}) \dx$.
It turns out that the Jensen-Shannon divergence is a $f$-divergence for the generator

\begin{equation}
f(x)=\frac{1}{2} \left( (x+1)\log\frac{2}{x+1} + x\log x \right).
\end{equation}
$f$-divergences preserve the information monotonicity~\cite{alphabetagamma-2010}, and their differential geometry was studied by Vos~\cite{vos-fdiv-1991}.
However, this  Jensen-Shannon divergence is a very particular case of Burbea-Rao divergences since the squared Euclidean distance (another Burbea-Rao divergence) does not belong to the class of $f$-divergences.

\section*{Source code}
The generic Burbea-Rao barycenter estimation algorithm shall be released in the \jMEF{} open source library:\\

\centerline{\url{http://www.informationgeometry.org/MEF/}}

An applet visualizing the skew Burbea-Rao centroids ranging from the right-sided to left-sided Bregman centroids is available at:
\centerline{\url{http://www.informationgeometry.org/BurbeaRao/}}

\section*{Acknowledgments}
We gratefully acknowledge financial support from French agency DIGITEO (GAS 2008-16D) and French National Research Agency (ANR GAIA 07-BLAN-0328-01), and Sony Computer Science Laboratories, Inc. 
We are very grateful to the reviewers for their thorough and thoughtful comments and suggestions.
In particular, we are thankful to the anonymous Referee that pointed out a rigorous proof of Lemma~1.

\appendix

\section*{Proof of uniqueness of the Burbea-Rao centroids\label{sec:proof}}

Consider without loss of generality the Burbea-Rao centroid (also called Jensen centroid) defined as the minimizer of

$$
c=\arg\min_x \sum_{i=1}^n \frac{1}{n} J_F(p_i,x)
$$
where $J_F(p,q)= \frac{F(p)+F(q)}{2} - F(\frac{p+q}{2}) \geq 0$.
For sake of simplicity, let us consider univariate generators.
The Jensen divergence may not be convex as 
$J'(x,p)= \frac{F'(x)}{2}-\frac{1}{2}F(\frac{x+p}{2})$ and 
$J''(x,p) =\frac{1}{2}F''(x) - \frac{1}{4}F''(\frac{x+p}{2})$ can be alternatively positive/negative (see~\cite{SkewJensen-2011}).
In general, minimizing the average non-convex divergence may {\it a priori} yield to many local minima~\cite{ManyLocalMinima-1995}.
It is remarkable to observe that the centroid induced by a Jensen divergence is unique although the problem may not be convex.

The proof of uniqueness of the Burbea-Rao centroid and the convergence of the CCCP approximation algorithm  rely on the ``interness'' property  (called compensativeness\footnote{A  fact following from the monotonicity of the generator function $\nabla F$.} in~\cite{Aggregator-1998}) of quasi-arithmetic means: 

$$
\min_{i=1}^n p_i \leq M_{\nabla F}(p_1, ..., p_n)  \leq \max_{i=1}^n p_i,
$$
with 
$$
M_{\nabla F}(p_1, ..., p_n)=(\nabla F)^{-1} \left(\sum_{i=1}^n \frac{1}{n} \nabla F(p_i) \right)
$$ 
for a strictly convex function $F$ (and hence, strictly monotone increasing gradient $\nabla F$).
The interness property of quasi-arithmetic means ensures that it is indeed a {\it mean} value contained within the extremal values.

For sake of simplicity, let us first consider a univariate convex generator $F$ with the $p_i$'s following the increasing order: 
$p_1\leq ... \leq p_n$.
Let initially $c_0\in [p_1^{(0)}=p_1, p_n^{(0)}=p_n]$.
Since $c_1=M_{\nabla F}( p_1^{(1)}=\frac{p_1+c_0}{2}, ..., p_n^{(1)}= \frac{p_n+c_0}{2})$ is a quasi-arithmetic mean, we necessarily have
 $c_1\in [p_1^{(1)}, p_n^{(1)}]$ and 
$p_n^{(1)}-p_1^{(1)} = \frac{c_0+p_n-c_0-p_1}{2} = \frac{p_n-p_1}{2}$.
Thus the CCCP iterations induce a sequence of iterated quasi-arithmetic means $c_t$ such that

\begin{eqnarray*}
c_t &=& M_{\nabla F}\left(p_1^{(t)}=\frac{p_1^{(t-1)}+c_{t-1}}{2}, ..., p_n^{(t)}=\frac{p_n^{(t-1)}+c_{t-1}}{2}\right), \\
c_t &\in& [p_1^{(t)}, p_n^{(t)}]
\end{eqnarray*}
with 
\begin{eqnarray*}
p_n^{(t)}-p_1^{(t)} &=& \frac{c_{t-1}+p_n^{(t-1)}-c_{t-1}-p_1^{(t-1)}}{2},\\
& =& \frac{p_n^{(t-1)}-p_1^{(t-1)}}{2},\\
&=&\frac{1}{2^t} (p_n-p_1).
\end{eqnarray*}

It follows that the sequence of centroid approximation $c_t$ converges in the limit to a unique  centroid $c^*$.
That is, the Burbea-Rao centroids exist and are unique for any strictly convex generator $F$.
The centroid can be approximated within $\frac{1}{2^t}$ relative precision after $t$ iterations (linear convergence of the CCCP).
Since the CCCP iterations yield both an approximation $c_t$ and a range $[p_1^{(t)}, p_n^{(t)}]$ where $c_t$ should be at the $t$-iteration, 
 we choose in practice to stop iterating whenever $\frac{p_n^{(t)}-p_1^{(t)}}{p_n^{(t)}}$ goes below a prescribed threshold 
(for example, taking $\nabla F=\log x$, we find in about $50$ iterations the centroid with machine precision $10^{-12}$).
The CCCP algorithm with $b$ bits precision require $O(nb)$ time to approximate.

Note that $\lim_{t\rightarrow\infty} p_1^{(t)} =\lim_{t\rightarrow\infty} \frac{1}{2}\sum_{i=0}^t \frac{1}{2^i} p_1 = p_1$ (and similarly, we have $\lim_{t\rightarrow\infty} p_n^{(t)}=p_n$).
It follows that $c^*\in [p_1,p_n]$ as expected (all the initial extremal range is possible, and the center shall depend on the chosen generator $F$).
The proof extends naturally to separable multivariate functions by carrying the analysis on each dimension independently.



\begin{IEEEbiography}{Frank Nielsen}
received the BSc (1992) and MSc (1994)
degrees from Ecole Normale Superieure (ENS Lyon, France).
He prepared  his PhD  on adaptive
computational geometry at INRIA
Sophia-Antipolis (France) and defended it in 1996. 
As a civil servant of the University of Nice (France), he
gave lectures at the engineering schools ESSI
and ISIA (Ecole des Mines). In 1997, he served
in the army as a scientific member in the computer science laboratory of
Ecole Polytechnique. In 1998, he joined Sony Computer Science
Laboratories Inc., Tokyo (Japan) where he is senior researcher. 
He became a professor of the CS Dept. of Ecole Polytechnique in 2008.
His current research
interests include geometry, vision, graphics, learning, and optimization.
He is a senior ACM and senior IEEE member.
\end{IEEEbiography}

\begin{IEEEbiography}{Sylvain Boltz}
Sylvain Boltz received the M.S. degree and the
Ph.D. degree in computer vision from the University
of Nice-Sophia Antipolis, France, in 2004 and 2008,
respectively.
Since then, he has been a postdoctoral fellow at
the VisionLab, University of California, Los Angeles and
a LIX-Qualcomm postdoctoral fellow in Ecole Polytechnique, France.
His research spans computer vision and image, video
processing with a particular interest in applications of
information theory and compressed sensing to these
areas.
\end{IEEEbiography}


\begin{thebibliography}{10}
\providecommand{\url}[1]{#1}
\csname url@samestyle\endcsname
\providecommand{\newblock}{\relax}
\providecommand{\bibinfo}[2]{#2}
\providecommand{\BIBentrySTDinterwordspacing}{\spaceskip=0pt\relax}
\providecommand{\BIBentryALTinterwordstretchfactor}{4}
\providecommand{\BIBentryALTinterwordspacing}{\spaceskip=\fontdimen2\font plus
\BIBentryALTinterwordstretchfactor\fontdimen3\font minus
  \fontdimen4\font\relax}
\providecommand{\BIBforeignlanguage}[2]{{%
\expandafter\ifx\csname l@#1\endcsname\relax
\typeout{** WARNING: IEEEtran.bst: No hyphenation pattern has been}%
\typeout{** loaded for the language `#1'. Using the pattern for}%
\typeout{** the default language instead.}%
\else
\language=\csname l@#1\endcsname
\fi
#2}}
\providecommand{\BIBdecl}{\relax}
\BIBdecl

\bibitem{mean-kolmogorov-1930}
A.~N. Kolmogorov, ``Sur la notion de la moyenne,'' \emph{Accad. Naz. Lincei
  Mem. Cl. Sci. Fis. Mat. Natur. Sez.}, vol.~12, pp. 388--391, 1930.

\bibitem{mean-nagumo-1930}
M.~Nagumo, ``{\"U}ber eine {K}lasse der {M}ittelwerte,'' \emph{Japanese Journal
  of Mathematics}, vol.~7, pp. 71--79, 1930, see Collected papers, Springer
  1993.

\bibitem{meanvalueaxiomatization-1948}
J.~D. Acz\'el, ``On mean values,'' \emph{Bulletin of the American Mathematical
  Society}, vol.~54, no.~4, pp. 392--400, 1948, http://www.mta.hu/.

\bibitem{entropicmeans-1989}
A.~Ben-Tal, A.~Charnes, and M.~Teboulle, ``Entropic means,'' \emph{Journal of
  Mathematical Analysis and Applications}, vol. 139, no.~2, pp. 537 -- 551,
  1989.

\bibitem{AliSilvey-1966}
S.~M. Ali and S.~D. Silvey, ``A general class of coefficients of divergence of
  one distribution from another,'' \emph{Journal of the Royal Statistical
  Society, Series B}, vol.~28, pp. 131--142, 1966.

\bibitem{Csiszar-1967}
I.~Csisz\'ar, ``Information-type measures of difference of probability
  distributions and indirect observation,'' \emph{Studia Scientiarum
  Mathematicarum Hungarica}, vol.~2, p. 229–318, 1967.

\bibitem{2009-BregmanCentroids-TIT}
F.~Nielsen and R.~Nock, ``Sided and symmetrized {B}regman centroids,''
  \emph{IEEE Transactions on Information Theory}, vol.~55, no.~6, pp.
  2048--2059, June 2009.

\bibitem{CensorZenios-1997}
Y.~Censor and S.~A. Zenios, \emph{Parallel Optimization: Theory, Algorithms,
  and Applications}.\hskip 1em plus 0.5em minus 0.4em\relax Oxford University
  Press, 1997.

\bibitem{WainwrightJordan-2008}
M.~J. Wainwright and M.~I. Jordan, ``Graphical models, exponential families,
  and variational inference,'' \emph{Foundational Trends in Machine Learning},
  vol.~1, pp. 1--305, January 2008.

\bibitem{alphaunique-2009}
S.-I. Amari, ``{$\alpha$}-divergence is unique, belonging to both
  $f$-divergence and bregman divergence classes,'' \emph{IEEE Trans. Inf.
  Theor.}, vol.~55, no.~11, pp. 4925--4931, 2009.

\bibitem{amari-2007}
S.-i. Amari, ``Integration of stochastic models by minimizing
  $\alpha$-divergence,'' \emph{Neural Comput.}, vol.~19, no.~10, pp.
  2780--2796, 2007.

\bibitem{informationgeometry-2000}
S.~Amari and H.~Nagaoka, \emph{Methods of Information Geometry}, A.~M. Society,
  Ed.\hskip 1em plus 0.5em minus 0.4em\relax Oxford University Press, 2000.

\bibitem{isvd2009}
F.~Nielsen and R.~Nock, ``The dual voronoi diagrams with respect to
  representational bregman divergences,'' in \emph{International Symposium on
  Voronoi Diagrams (ISVD)}.\hskip 1em plus 0.5em minus 0.4em\relax DTU Lyngby,
  Denmark: IEEE, June 2009.

\bibitem{BregmanVoronoi-2010}
J.-D. Boissonnat, F.~Nielsen, and R.~Nock, ``Bregman {V}oronoi diagrams,''
  \emph{Discrete \& Computational Geometry}, 2010, accepted, extend ACM-SIAM
  SODA 2007.

\bibitem{meangenproj-1995}
I.~Csisz\'ar, ``Generalized projections for non-negative functions,''
  \emph{Acta Mathematica Hungarica}, vol.~68, no. 1-2, pp. 161--185, 1995.

\bibitem{bregmankmeans-2005}
A.~Banerjee, S.~Merugu, I.~S. Dhillon, and J.~Ghosh, ``Clustering with
  {B}regman divergences,'' \emph{J. Mach. Learn. Res.}, vol.~6, pp. 1705--1749,
  2005.

\bibitem{Detyniecki-mathematicalaggregation-2000}
M.~Detyniecki, ``Mathematical aggregation operators and their application to
  video querying,'' Ph.D. dissertation, 2000.

\bibitem{BurbeaRao-1982}
J.~Burbea and C.~R. Rao, ``On the convexity of some divergence measures based
  on entropy functions,'' \emph{IEEE Transactions on Information Theory},
  vol.~28, no.~3, pp. 489--495, 1982.

\bibitem{BurbeaRao-higherorder-1982}
------, ``On the convexity of higher order {J}ensen differences based on
  entropy functions,'' \emph{IEEE Transactions on Information Theory}, vol.~28,
  no.~6, pp. 961--, 1982.

\bibitem{Jensen-Shannon-divergence}
J.~Lin, ``Divergence measures based on the {S}hannon entropy,'' \emph{IEEE
  Transactions on Information Theory}, vol.~37, pp. 145--151, 1991.

\bibitem{entdivmean-1995}
M.~Basseville and J.-F. Cardoso, ``On entropies, divergences and mean values,''
  in \emph{Proceedings of the IEEE Workshop on Information Theory}, 1995.

\bibitem{Bregman67}
L.~M. Bregman, ``The relaxation method of finding the common point of convex
  sets and its application to the solution of problems in convex programming,''
  \emph{USSR Computational Mathematics and Mathematical Physics}, vol.~7, pp.
  200--217, 1967.

\bibitem{ef-flashcards-2009}
F.~Nielsen and V.~Garcia, ``Statistical exponential families: A digest with
  flash cards,'' 2009, arXiv.org:0911.4863.

\bibitem{Fukunaga90}
K.~Fukunaga, \emph{Introduction to statistical pattern recognition (2nd
  ed.)}.\hskip 1em plus 0.5em minus 0.4em\relax Academic Press Professional,
  Inc., 1990.

\bibitem{BattGaussian}
\BIBentryALTinterwordspacing
L.~Rigazio, B.~Tsakam, and J.~C. Junqua, ``An optimal bhattacharyya centroid
  algorithm for gaussian clustering with applications in automatic speech
  recognition,'' in \emph{Acoustics, Speech, and Signal Processing, 2000.
  ICASSP '00. Proceedings. 2000 IEEE International Conference on}, vol.~3,
  2000, pp. 1599--1602 vol.3. [Online]. Available:
  \url{http://ieeexplore.ieee.org/xpls/abs_all.jsp?arnumber=861998}
\BIBentrySTDinterwordspacing

\bibitem{2010-Hierachical-ICCASP}
V.~Garcia, F.~Nielsen, and R.~Nock, ``Hierarchical gaussian mixture model,'' in
  \emph{International Conference on Acoustics, Speech and Signal Processing
  (ICASSP)}, March 2010.

\bibitem{metricbregman1-2008}
P.~Chen, Y.~Chen, and M.~Rao, ``Metrics defined by {B}regman divergences: Part
  {I},'' \emph{Commun. Math. Sci.}, vol.~6, pp. 9915--926, 2008.

\bibitem{metricbregman2-2008}
------, ``Metrics defined by {B}regman divergences: Part {II},'' \emph{Commun.
  Math. Sci.}, vol.~6, pp. 927--948, 2008.

\bibitem{Jeffreys46}
H.~Jeffreys, ``An invariant form for the prior probability in estimation
  problems,'' \emph{Proceedings of the Royal Society of London}, vol. 186, no.
  1007, pp. 453--461, March 1946.

\bibitem{LieseVajda-2006}
F.~Liese and I.~Vajda, ``{On Divergences and Informations in Statistics and
  Information Theory},'' \emph{IEEE Transactions on Information Theory},
  vol.~52, no.~10, pp. 4394--4412, October 2006.

\bibitem{zhang-2004}
J.~Zhang, ``Divergence function, duality, and convex analysis,'' \emph{Neural
  Computation}, vol.~16, no.~1, pp. 159--195, 2004.

\bibitem{YuilleNC03}
A.~Yuille and A.~Rangarajan, ``{The concave-convex procedure},'' \emph{Neural
  Computation}, vol.~15, no.~4, pp. 915--936, 2003.

\bibitem{SriperumbudurNIPS09}
B.~Sriperumbudur and G.~Lanckriet, ``On the convergence of the concave-convex
  procedure,'' in \emph{Neural Information Processing Systems}, 2009.

\bibitem{jensenrenyi-isar-2001}
Y.~He, A.~B. Hamza, and H.~Krim, ``An information divergence measure for {ISAR}
  image registration,'' in \emph{Automatic target recognition XI (SPIE)}, vol.
  4379, 2001, pp. 199--208.

\bibitem{JensenRenyi-CSPL328-2001}
A.~O. Hero, B.~Ma, O.~Michel, and J.~D. Gorman, ``Alpha-divergence for
  classification, indexing and retrieval,'' Comm. and Sig. Proc. Lab. (CSPL),
  Dept. EECS, University of Michigan, Ann Arbor, Tech. Rep. 328, July, 2001,
  presented at Joint Statistical Meeting.

\bibitem{Bhatta1943}
A.~Bhattacharyya, ``On a measure of divergence between two statistical
  populations defined by their probability distributions,'' \emph{Bulletin of
  Calcutta Mathematical Society}, vol.~35, pp. 99--110, 1943.

\bibitem{Matusita-1955}
K.~Matusita, ``Decision rules based on the distance, for problems of fit, two
  samples, and estimation,'' \emph{Annal of Mathematics and Statistics},
  vol.~26, pp. 631--640, 1955.

\bibitem{KailathITCT67}
T.~Kailath, ``{The divergence and Bhattacharyya distance measures in signal
  selection},'' \emph{IEEE Transactions on Communication Technology}, vol.~15,
  no.~1, pp. 52--60, 1967.

\bibitem{Bhat1998}
F.~Aherne, N.~Thacker, and P.~Rockett, ``The {B}hattacharyya metric as an
  absolute similarity measure for frequency coded data,'' \emph{Kybernetika},
  vol.~34, no.~4, pp. 363--368, 1998.

\bibitem{hellinger-1907}
E.~D. Hellinger, ``Die orthogonalinvarianten quadratischer formen von unendlich
  vielen variablen,'' 1907, thesis of the university of G\"ottingen.

\bibitem{kakutani-1948}
S.~Kakutani, ``On equivalence of infinite product measures,'' \emph{Annals of
  Mathematics}, vol.~49, no. 214-224, 1948.

\bibitem{RigazioICASSP00}
L.~Rigazio, B.~Tsakam, and J.~Junqua, ``{Optimal Bhattacharyya centroid
  algorithm for Gaussian clustering with applications in automatic speech
  recognition},'' in \emph{IEEE International Conference on Acoustics, Speech,
  and Signal Processing}, vol.~3, 2000, pp. 1599--1602.

\bibitem{Petersen04}
\BIBentryALTinterwordspacing
K.~B. Petersen and M.~S. Pedersen, \emph{The Matrix Cookbook}.\hskip 1em plus
  0.5em minus 0.4em\relax Technical University of Denmark, oct 2008. [Online].
  Available: \url{http://www2.imm.dtu.dk/pubdb/p.php?3274}
\BIBentrySTDinterwordspacing

\bibitem{alphabetagamma-2010}
A.~Cichocki and S.~ichi Amari, ``Families of alpha- beta- and gamma-
  divergences: Flexible and robust measures of similarities,'' \emph{Entropy},
  2010, review submitted.

\bibitem{vos-fdiv-1991}
P.~Vos, ``Geometry of $f$-divergence,'' \emph{Annals of the Institute of
  Statistical Mathematics}, vol.~43, no.~3, pp. 515--537, 1991.
  
\bibitem{SkewJensen-2011}  
F.~Nielsen, R.~Nock,
``Skew Jensen-Bregman Voronoi Diagrams.''
Transactions on Computational Science, no.~14, pp. 102--128, 2011.

\bibitem{ManyLocalMinima-1995}
P.~Auer, M.~Herbster and M.~Warmuth, 
``Exponentially many local minima for single neurons,''
Advances in Neural Information Processing Systems, vol. 8, pp. 316–-317, 1995. 

\bibitem{Aggregator-1998}
J.-L. Marichal,
``Aggregation Operators for Multicriteria Decision Aid,''
Institute of Mathematics, University of Li\`ege,   Belgium, 1998. 
  
\end{thebibliography}
\end{document}